\documentclass[letterpaper, 10 pt, conference]{ieeeconf}  %

\IEEEoverridecommandlockouts                              %
\usepackage{times}
\usepackage{booktabs}  
\usepackage{xcolor}
\usepackage{amsmath,amssymb,amsfonts}
\usepackage{graphicx}
\usepackage{makecell}
\usepackage{multirow}
\usepackage{mathrsfs}
\usepackage{cite}
\usepackage{dsfont}
\usepackage{url}
\usepackage{algorithmic}
\usepackage[ruled,vlined,linesnumbered]{algorithm2e}

\newtheorem{thm}{\bf Theorem}[section]

\newtheorem{rem}[thm]{\bf Remark}
\newtheorem{ass}[thm]{\bf Assumption}
\newtheorem{prop}[thm]{\bf Proposition}

\newtheorem{deff}[thm]{\bf Definition}

\newcommand{\tf}{\tilde{f}}
\newcommand{\hf}{\hat{f}}

\newcommand{\X}{\mathcal{X}}
\newcommand{\Y}{\mathcal{Y}}
\renewcommand{\ll}{\mathcal{L}}
\renewcommand{\lll}{\mathcal{L}_{\lambda}}

\newcommand{\D}{\mathcal{D}}

\newcommand{\rr}{\mathcal{R}}

\newcommand{\K}{\mathcal{K}}

\renewcommand{\O}{\mathcal{O}}
\renewcommand{\Sigma}{C}

\newcommand{\R}{\mathbb{R}}
\newcommand{\N}{\mathbb{N}}

\newcommand{\Lb}{L}
\newcommand{\Xb}{X}

\newcommand{\xb}{x}

\newcommand{\Zk}{\mathcal{Z}}
\newcommand{\zk}{\mathfrak{z}}

\newcommand{\eps}{\varepsilon}

\newcommand{\ra}{\rightarrow}

\newcommand{\id}{\operatorname{Id}}

\newcommand{\norm}[1]{\left\Vert #1 \right\Vert}

\newcommand{\taus}{{\tau_s}}
\newcommand{\trans}{{\mathsf{T}}}

\title{\LARGE \bf
Learning Koopman-based Stability Certificates for Unknown Nonlinear Systems
}

\author{Ruikun Zhou, Yiming Meng, Zhexuan Zeng, and Jun Liu, \IEEEmembership{Senior Member, IEEE}
\thanks{This research was supported in part by an NSERC Discover Grant and the Canada Research Chairs program.}%
\thanks{Ruikun Zhou and Jun Liu are with the Department of Applied Mathematics, Faculty of Mathematics, University of Waterloo, Waterloo, Ontario N2L 3G1, Canada.  Emails: \texttt{ruikun.zhou@uwaterloo.ca, j.liu@uwaterloo.ca}}%
\thanks{Yiming Meng is with  the
Coordinated Science Laboratory, University of Illinois Urbana-Champaign,
Urbana, IL 61801, USA. Email: 
        \texttt{ymmeng@illinois.edu}}%
\thanks{Zhexuan Zeng is with Department of Automatic Control, Huazhong University of Science and Technology, Wuhan, China. Email: 
        \texttt{xuanxuan@hust.edu.cn}}%
}

\begin{document}

\maketitle
\thispagestyle{empty}
\pagestyle{empty}

\begin{abstract}
 Koopman operator theory has gained significant attention in recent years for identifying discrete-time nonlinear systems by embedding them into an infinite-dimensional linear vector space. However, providing stability guarantees while learning the continuous-time dynamics, especially under conditions of relatively low observation frequency, remains a challenge within the existing Koopman-based learning frameworks. To address this challenge, we propose an algorithmic framework to simultaneously learn the vector field and Lyapunov functions for unknown nonlinear systems, using a limited amount of data sampled across the state space and along the trajectories at a relatively low sampling frequency. The proposed framework builds upon recently developed high-accuracy Koopman generator learning for capturing transient system transitions and physics-informed neural networks for training Lyapunov functions. We show that the learned Lyapunov functions can be formally verified using a satisfiability modulo theories (SMT) solver and provide less conservative estimates of the region of attraction compared to existing methods.
\end{abstract}


\section{Introduction}
Stability properties can be qualitatively characterized by Lyapunov functions for various nonlinear systems. However, discovering a Lyapunov function remains a long-standing challenge in nonlinear dynamical systems, even for known systems. For unknown systems, this challenge is further compounded by the accuracy of data-driven system identification and the formal verification of the learned Lyapunov functions. Below, we review the recent advances and challenges associated with these two factors before proposing the learning structure for stability certificates.

The development of Koopman operator theory~\cite{koopman1931hamiltonian} for dynamical systems provides a promising alternative learning approach for nonlinear system identification and stability analysis~\cite{mauroy2020koopman}, using data from snapshots of the flow map (trajectories). By leveraging the spectral properties~\cite{mezic_spectral_2005} and mode decomposition capabilities~\cite{schmid2010dynamic,williams2015data} of Koopman operators, nonlinear dynamics can be converted into a discrete family of observable functions (functions of the states), which evolve linearly in the infinite-dimensional function space driven by the Koopman operator. This nonlinear-to-linear conversion, seemingly friendly at first, is later found to require the study of continuous functional calculus. %
It suffers from restrictive assumptions on the function properties~\cite{lusch_deep_2018} to ensure correct usage. Consequently, many applications are limited to discrete-time nonlinear systems.

In identifying the continuous-time dynamics of an unknown system, \cite{mauroy2019koopman} proposed a method using the system's (infinitesimal) generator, also known as the Koopman generator. This approach avoids the need for time derivatives, requiring fewer data and outperforming widely used methods at low sampling frequencies, such as the sparse identification of nonlinear dynamics (SINDy) technique~\cite{brunton2016discovering, klus2020data}, which heavily relies on low-resolution finite-difference approximations of time derivatives from discrete-time sampled snapshots. However, the method in \cite{mauroy2019koopman} requires the desirable diagonal property of the Koopman operator, which is typically difficult to satisfy for most unknown nonlinear systems, thus limiting its correct usage. This limitation has been addressed by recent research in \cite{meng2024koopman}, which   improves learning reliability and accuracy.

Koopman operators also offer a powerful tool for constructing Lyapunov functions and stability analysis~\cite{mauroy_global_2016,yi2023equivalence}. %
In~\cite{mauroy_global_2016}, the authors demonstrated that a set of Lyapunov functions for nonlinear systems with global stability can be constructed from the eigenfunctions of the Koopman operator. Building on this result, and incorporating the autoencoder structure~\cite{azencot2020forecasting} in Koopman operator learning, the authors of~\cite{deka2022koopman} introduced an algorithm to identify  Koopman eigenfunctions that parameterize a set of Lyapunov function candidates. 
More recently, \cite{umathe2023spectral} used the spectrum of the Koopman operator to directly identify the stability boundary of nonlinear systems. 
However, due to the low-resolution quantification of the Lyapunov candidate function space, as quantified by the Lyapunov derivative inequality, all the aforementioned stability verification methods face the limitation of conservative estimation of the region of attraction (ROA). In contrast, \cite{meng2023learning} proposed a modified Zubov-Koopman operator approach to approximate a maximal Lyapunov function, which is related to the solution of stability-related PDEs (the Lyapunov or Zubov equation \cite{liu2023physics}), where the non-trivial domain corresponds to the ROA. Regardless, all existing Koopman-based constructions of Lyapunov functions cannot be easily integrated into the system identification structure, and therefore lack the ability to be formally verified using the predicted system vector field, leading to a loss of the true predictability of the certifiable region of attraction.

In this paper, we propose an integrated algorithmic framework to simultaneously learn the vector field and Lyapunov functions with formal guarantees for unknown nonlinear systems, thereby achieving full predictability of stability. Our methodology leverages the strengths of the Koopman generator learning framework and utilizes the learned generator to approximately solve stability-related PDEs, with the aim of enhancing the formally verified ROA. In light of the recent development of physics-informed neural network solutions for PDEs, the proposed method in this paper is not merely a straightforward combination of Koopman generator learning and PDE solving using a parameterized ansatz to match the data and the equation. The input of human knowledge on the physics-informed conditions and the corresponding design of the loss function to train the solution are essential components.

It is worth noting that the technique proposed in~\cite{meng2023learning} can learn a near-maximal Lyapunov candidate for unknown systems, but requires a computationally intensive deep neural network smoothing process. Moreover, the Lyapunov derivative of the candidate cannot be verified due to the lack of information about the system's vector field. On the other hand, a non-Koopman neural system learning and Lyapunov construction framework~\cite{zhou2022neural} can provide formal guarantees to ensure closed-loop stability. However, this approach does not offer insights into the quality of the resulting region of attraction (ROA). We address these shortcomings and, to the best of the authors' knowledge, this is the first work to develop a systematic approach for finding formally verified Lyapunov functions using the Koopman generator. 

The detailed contributions are as follows.
\begin{itemize}
    \item We propose a streamlined framework to simultaneously learn the vector field and Lyapunov functions by reusing the same observable test functions and a subset of the data samples.
     \item  Beyond formulating the least-squares problem, we provide additional conditions for solving the stability-related PDEs using the Koopman generator.
     \item We demonstrate that the learned Lyapunov function is valid for unknown systems, with formal guarantees achieved through verification using SMT solvers.
 \end{itemize}

\section{Preliminaries}

Throughout this paper, we consider a continuous-time nonlinear dynamical system of the form 

\begin{equation}\label{eq: sys}
   \dot{x} = f(x), \quad x(0) = x_0, 
\end{equation}
and the vector field  $f:\X\ra\R^n$ is assumed to be locally Lipschitz continuous, where $\X \subseteq \R^n$ is a pre-compact state space.  For each initial condition $x_0$, we denote the unique forward flow map, i.e., the solution map, by $\phi: I\times \X\rightarrow \X$, where $I$ is the maximal interval of existence. In this paper, we assume $I=[0,\infty)$ to the initial value problem  \eqref{eq: sys}. The flow map should then satisfy 1) $\partial_t(\phi(t,x)) = f(\phi(t,x))$, 2) $\phi(0, x)=x_0$, and 3) $\phi(s, \phi(t,x))=\phi(t+s, x)$ for all $t,s\in I$.

\subsection{Koopman Operators and the Infinitesimal Generator}
\begin{deff}%
    \label{eq_koopman} %
The Koopman operator family $\{\K_t\}_{t\geq 0}$ of system \eqref{eq: sys} is a collection of maps $\mathcal{K}_t: \mathcal{C}^1(\X) \rightarrow \mathcal{C}^1(\X)$  defined by
\begin{align} \label{eq: Kt}
\mathcal{K}_t h = h \circ \phi(t, \cdot), \quad h \in \mathcal{C}^1(\X)
\end{align}
for each $t\geq 0$, where $\circ$ is the composition operator. The (infinitesimal) generator $\ll$ of $\{\K_t\}_{t\geq 0}$, which is also called the Koopman generator, is defined as $\ll h:= \lim_{t\ra 0}\frac{\K_t h-h}{t}$ and equivalent to
\begin{equation} \label{eq:equality}
    \ll h(x) = \nabla h(x) \cdot f(x), \;\;h\in\mathcal{C}^1(\X). 
\end{equation}
\end{deff}

It is known that $\{\K_t\}_{t\geq 0}$ forms a linear $\mathcal{C}_0$-semigroup~\cite{meng2024resolvent}, and thus there exist constants $\omega\geq 0$ (indicating the exponential growth rate) and $\Sigma\geq 1$ (representing a spatial uniform scaling) such that $\norm{\K_t}\leq C e^{\omega t}, \; \forall t\geq 0$~\cite{pazy2012semigroups}.

\subsection{Concept of Stability}
Without loss of generality, we assume the origin is an equilibrium point of the system~\eqref{eq: sys}. 
\begin{deff}[Asymptotic Stability]
The origin is
said to be asymptotically stable for \eqref{eq: sys}
if 1) for every $\eps>0$, there exists  a $\delta>0$ such that $|x| < \delta$ implies $|\phi(t, x)| <\eps$ for all $t\geq 0$, and 2) there exists a neighborhood $U\subseteq \X$ of the origin such that, for each $x_0 \in U$, we have $\phi(t, x) \in U$ for all $t\geq 0$ and $\lim_{t\ra\infty}|\phi(t, x)|=0$.
\end{deff}
We further define the region of attraction of the equilibrium point given its asymptotic stability. 
\begin{deff}[ROA]
Suppose that the origin is asymptotically stable. The domain of attraction of it is a  set  defined as
$ \D :=\{x\in\X: \lim_{t\ra\infty} |\phi(t, x)|=0\}$.
Any forward invariant subset $\rr \subseteq \D$ is called a region of attraction (ROA).
\end{deff}

\begin{thm}[Lyapunov Theorem]
\label{thm: lyapunov}
Let $D\subseteq \R^n$ be an open set containing the origin. Consider the nonlinear system~\eqref{eq: sys}. Let $V(x): D \ra \R$ be a continuously differentiable function such that 
\begin{align}
\begin{split}
        V(\textbf{0})=0 \text { and } V(x)>0 & \text { for all } x \in D \setminus\{\textbf{0}\}, 
        \\ 
        \mathcal{L}V(x):=\nabla V(x)\cdot f(x)<0 & \text { for all } x \in D \setminus\{\textbf{0}\},  
\end{split}
\end{align}
Then, the origin is asymptotically stable for the system.
\end{thm}

An immediate consequence of Theorem \ref{thm: lyapunov} (and its proof \cite{khalil_nonlinear_2015}) is that the sub-level set $V^c = \{x\in D \mid V(x) \leq c\}$ is an ROA for the system~\eqref{eq: sys}, where $c > 0$ is such that $V^c$ does not intersect with the boundary of $D$. 

 Given a nonlinear system is asymptotically stable, a Lyapunov function exists by the converse Lyapunov theorem~\cite{khalil_nonlinear_2015}. But converse Lyapunov functions are usually nonconstructive. Despite their fundamental importance, computing Lyapunov functions remains a challenging task. Existing methods, such as sum-of-squares (SOS) and neural Lyapunov functions, typically yield local results with highly conservative estimates of the ROA, as their searching policies are guided by the Lyapunov derivative inequality, which admits non-unique solutions. In contrast, 
as demonstrated in~\cite{liu2023physics, zhou2024physics} and the recent literature, the maximal Lyapunov function~\cite{vannelli1985maximal}, which provides the largest estimation of ROA, can be characterized by the following partial differential equation (PDE):
    \begin{equation}\label{eq: lyapunov}
   \mathcal{L}V(x)   = - \eta(x),  
     \end{equation}
which we refer to the Lyapunov equation. Here $\eta$ is a positive definite function over $\X$ with respect to the origin. However, due to the unbounded nature of the Lyapunov function, it remains challenging to approximate the true domain of attraction using typical data-driven learning methods, given the compactness of the sampling subspace. The following theorem states that the domain of attraction can be characterized by a transformed maximal Lyapunov function with a bounded range from $0$ to $1$, where the $1$-sublevel set represents the domain of attraction.
\begin{thm}[Zubov's Theorem~\cite{zubov1964methods}]
\label{thm: zubov}
Let $D\subset\R^n$ be an open set containing the origin. Then $D=\D$ if and only if there exist two continuous functions $W:\, D \ra \R$ and $\eta:\, D \ra \R$ such that the following conditions hold: 1) $0<W(x)<1$ for all $x\in D\setminus \{\textbf{0}\}$ and $W(\textbf{0})=0$;  2) $ \eta$ is positive definite on $D$ with respect to $0$;   3) for any sufficiently small $c_3>0$, there exist $c_1, c_2>0$ such that $|x| \ge c_3$ implies 
     $W(x)>c_1$ and $\eta(x)>c_2$; 4) $W(x)\ra 1$ as $x\ra y$ for any $y\in \partial D$; 5) $W$ and $\eta$ satisfy
    \begin{equation}\label{eq: zubov}
    \mathcal{L}W(x) +\eta(x)(1-W(x)) = 0.
     \end{equation}
\end{thm}
We refer to \eqref{eq: zubov} as the Zubov equation in the following context.

\subsection{Problem Formulation}
The idea of this paper is to identify the Koopman generator and then solve the stability-related PDEs outlined above, with the goal of enlarging the estimation of the ROA. In previous work, \cite{liu2023physics, zhou2024physics} demonstrated excellent performance in learning accuracy. However, when solving PDEs with learned generators for unknown systems,
a simple data-fitting strategy by minimizing the residual is typically insufficient, as additional physical information—such as boundary conditions—must also be incorporated. In the following section, we first recap the resolvent-type method for learning the Koopman generator in~\cite{meng2024resolvent}. In Section~\ref{sec:data-driven}, we focus on developing the physics-informed Koopman-generator-based construction of the Lyapunov function to address the technical concerns mentioned above.

\section{Learning the Generator and Identifying the Vector Field via Koopman Resolvent}
\label{sec: vector field}
In this section, we recap the approximation of the Koopman generator and show that the generator can be approximated by the Koopman resolvent~\cite{susuki2021koopman}.

For $h \in \mathcal{C}^1(\X)$, we define the following integral representation of the Koopman resolvent: $\rr(\lambda; \ll)h:=\int_0^\infty e^{-\lambda t}(\K_t h) dt. $
Consequently, the Yosida approximation of the Koopman generator $\ll$ can be defined as $\lll =\lambda^2\rr(\lambda; \ll) -\lambda\id$, where $\id$ is the identity operator~\cite{pazy2012semigroups}. We have that $\lll$ converges to $\ll$ strongly as $\lambda \ra \infty$ .
In order to have a finite time approximation, we define a truncation integral operator
as 
\begin{equation} \label{eq:resolvent_truncated}
    \rr_{\lambda,\taus}  h(x):=\int_0^{\tau_s}e^{-\lambda s}\K_{s} h(x) ds.
\end{equation}
 Then, when $\lambda$ is sufficiently large, we should have an accurate approximation of the Koopman generator, as stated in the following theorem.

\begin{thm}[\cite{meng2024resolvent}]
\label{thm:error_bound}
    Define $\ll_{\lambda,\taus}:=  \lambda^2 \rr_{\lambda,\taus}  -\lambda\id$, and let $\taus\geq 0$ and $\lambda>\omega$ be fixed. Then, $ \| \ll_{\lambda,\taus}  -\lll \|\leq  \frac{\Sigma\lambda^2}{\lambda-\omega}e^{-\lambda \taus}$ on $\mathcal{C}^1(\X)$.
\end{thm}
We refer the readers to~\cite{meng2024resolvent} for detailed derivation and proofs. 
However, when $\lambda$ is large, it poses a challenge as it requires high sampling frequencies in practical applications. We then leverage the first resolvent identity $[(\lambda-\mu)\rr(\mu; \ll)+\id]\rr(\lambda; \ll) = \rr(\mu; \ll)$ to tackle this issue. By computing a finite-time horizon approximation $\rr_{\mu,\taus}$ for $\rr(\mu; \ll)$ first, $\ll_{\lambda,\taus}$ can be computed using this identity, where $\mu$ is a small positive number in the resolvent set of the Koopman generator. 

By \eqref{eq:equality} with $x := [x_1, x_2, \dots, x_n ]^\intercal \in \R^n$, we have $\ll x_i = f_i(x)$, for $i = 1, 2, \dots n$, where $f_i$ is defined as $\dot{x}_i = f_i(x)$, a component of $f$. As a result of Theorem~\ref{thm:error_bound}, when $x_1, x_2, \dots, x_n $ are included in the observable test functions, the vector field can be readily approximated:
\begin{equation}
\label{eq: f_hat}
    \hat{f}_i(x) = \ll_{\lambda,\taus} x_i \approx \ll x_i = f_i(x), \quad \forall i=1,2\cdots,n.
\end{equation}
We also direct the readers to ~\cite{meng2024resolvent} to see a set of numerical results on the effectiveness of \eqref{eq: f_hat} on the system identification tasks, varying from polynomial systems to general nonlinear systems.

\section{Stability Certificates for the Unknown Systems} \label{sec:stability}
In this section, we show that the Koopman generator can be used not only to approximate the vector field but also to establish stability certificates for unknown systems, i.e., constructing Lyapunov functions by solving linear PDEs. Furthermore, we show that the derived Lyapunov functions provide stability guarantees for the unknown systems.

We assume that the origin is a stable equilibrium point of~\eqref{eq: sys}. Then \eqref{eq: sys} can be rewritten as $\dot x = Ax + g(x),$
where $g(x)=f(x)-Ax$ satisfies $\lim_{x\ra \textbf{0}}\frac{\norm{g(x)}}{\norm{x}}=0$. If $A$ is known,  a local quadratic Lyapunov function $V_P(x)=x^TPx$ can be computed by solving $PA +A^T P = -Q$, where $Q$ is any symmetric positive definite matrix. It can be easily shown that there exists a sufficiently small $c>0$ such that $\O = \{x\in \X| V_P(x)\le c\}$, where the linearization dominates, that is a ROA~\cite[Section 5.1]{liu2023physics}. We make the following assumption.

\begin{ass} \label{ass: linearization}
 Let the unknown system be identified using~\eqref{eq: f_hat} and denote the linearization of the approximated vector field about the origin as $\dot{x} = \hat{A}x$. With $\hat{A}$ being Hurwitz (which can be checked numerically), let $P$ solve the Lyapunov euqation $P\hat{A} + \hat{A}^T P = -Q$ for some positive definite $Q$. We assume that identification is sufficiently accurate such that $\O = \{x\in \X| V_P(x)\le c\}$ is a ROA for (\ref{eq: sys}) for some $c>0$. 
\end{ass}

Note that this assumption can typically be satisfied in most cases by collecting a sufficient amount of data around the origin and ensuring $f(\textbf{0}) = \hf(\textbf{0})$ in practice.
Now, consider an unsampled state $z$, and $y$ represents the nearest known sample for which $f(y)$ and $\hf(y)$ are accessible. The following theorem establishes the stability conditions for the unknown systems.

\begin{prop} \label{prop:unknown_stability}
    Suppose that Assumption~\ref{ass: linearization} holds. Let $V:\,\X\ra \R$ be a candidate Lyapunov function for $\dot x = \hf(x)$, with $\hf(\textbf{0}) = 0$, and let $K_f$ and $K_{\hf}$ be the Lipschitz constants of $f$ and $\hf$ respectively, on $\X$. Let $\Y\subset \X$ be a finite set of samples. Let $\delta$, $\nu$, and $\alpha$ be positive constants such that 1) for every $z\in \X$, there exists $y\in \Y$ such that $\|z - y\| < \delta$, 2) $\sup_{x\in \X}\| \nabla V(x)\| \leq \nu$, and 3) $\max_{y\in \Y}\|f(y)-\hf(y)\|=\alpha$. Let $\beta$ be such that $((K_f+ K_{\hf}) \delta + \alpha) \nu < \beta $ and
    \begin{equation}
    \label{eq: relaxed V}
        \nabla V(x) \cdot \hf(x) \leq -\beta
    \end{equation}
    holds on $\Omega_{c_1,c_2} = \{x\in \X:\, c_1\le V(x)\le c_2\}$, for some $0<c_1<c_2$, such that 
    $\Omega_{c_1}=\{x\in \X:\, V(x)\le c_1\}\subset \O$, then $\Omega_{c_2}=\{x\in \X:\, V(x)\le c_2\}$ is an ROA for the unknown system (\ref{eq: sys}), provided that $\Omega_{c_2}$ does not  intersect with the boundary of $\X$.
\end{prop}
\begin{proof}
    For any $x \in \X$ and $y\in \Y$, we have 
    \begin{align}
        \begin{split}
            & \|f(x) - \hf(x)\| \\
            &  \leq \|f(x) - f(y) \|  + \|f(y) - \hf(y) \| + \|\hf(y) - \hf(x)\|  \\
            & \leq K_f \delta + \alpha + K_{\hf} \delta.
        \end{split}
    \end{align}
It follows that
    \begin{equation} 
    \label{eq: beta}
	\begin{split}
 &\nabla V(x) \cdot f(x) - \nabla V(x) \cdot \hf(x) \\
   & \leq \| \nabla V  \| \| f(x) - \hf(x)\| \\
   &  \leq \nu ((K_f+ K_{\hf}) \delta + \alpha)  < \beta. 
	\end{split}
\end{equation}
By (\ref{eq: relaxed V}) and (\ref{eq: beta}), we have 
    $
    \nabla V(x) \cdot f(x) < \nabla V(x) \cdot \hf(x) + \beta <0
    $ 
on $\Omega_{c_1,c_2}$. 
Hence all the trajectories in $\Omega_{c_2}$ converge to $\Omega_{c_1}\subset\O$. By Assumption~\ref{ass: linearization}, solutions of (\ref{eq: sys}) will converge to the origin from $\O$.
\end{proof}

\section{Data-Driven Algorithm for Deriving Stability Certificates}
\label{sec:data-driven}
In this section, we continue to discuss the data-driven learning algorithm based on the derivation for learning the vector field and Lyapunov function in the previous two sections.

\subsection{Learning the Koopman Generator}

First, we select a finite dictionary of continuously differentiable observable test functions, denoted by $    \Zk_N(x):=[\zk_{\scriptscriptstyle  1}(x), \zk_{\scriptscriptstyle  2}(x),  \cdots, \zk_{\scriptscriptstyle  N}(x)], \;N\in\N.$
Suppose that there exists a data-driven finite-dimensional approximation of $\ll_{\lambda,\taus}$, denoted as $L$, such that for any $h\in\operatorname{span}\{\zk_1, \zk_2, \cdots,\zk_N\}$, 
i.e., $h(x)=\Zk_N(x) \zeta$ where $\zeta$ is a column vector, we have that $\ll h(\cdot)\approx \ll_{\lambda,\taus} h(\cdot) \approx \Zk_N(\cdot)(\Lb\zeta) $.

By randomly sampling $M$ initial conditions $\{x^{(m)}\}_{m=1}^{M}\subseteq\X$ and fixing a $\taus$ with a sampling rate of $\gamma$ Hz along the trajectories, we construct the observable functions into $\Xb$, as $B =[\Zk_N(x^{(1)}), \Zk_N(x^{(2)}),\cdots, \Zk_N(x^{(M)})]^\trans.$ Under the specified sampling rate, to avoid inaccurate computation of the truncated integral of the Koopman resolvent~\eqref{eq:resolvent_truncated} for large $\lambda$, we first apply the Gauss–Legendre quadrature method~\cite{press2007numerical} using trajectory samples to obtain the numerical  approximation $\hat{\rr}_{\mu,\taus}$ for $\rr_{\mu,\taus}$ with a relatively small constant $\mu > 0$.
Subsequently, the generator $L$  can be inferred as 
\begin{equation} \label{eq:generator}
    L = X^\dagger Y
\end{equation}
where $X = (\lambda-\mu) \hat{\rr}_{\mu,\taus} + B$, and $Y = \lambda \mu \hat{\rr}_{\mu,\taus} -\lambda B$.
We direct the readers to \cite{meng2024resolvent} for a detailed algorithm and explanation of this resolvent-type method. 

By including $x$ in the dictionary, the approximated vector field $\hat{f}(x) \approx B\Zk_N(x)^T$, for some $B \in \R^{n \times N}$. To ensure that the approximated dynamics preserve the same equilibrium point at the origin, we can perform $\tilde{f}(x) = \hat{f}(x) - \hat{f}(\textbf{0})$, as demonstrated in~\cite{quartz2024stochastic}. In doing so, Assumption~\ref{ass: linearization} is easier to be satisfied with $\tilde{f}(\textbf{0}) =0 $. 
Again, we would like to emphasize that with this proposed method, the vector field is learned via learning the Koopman generator using the discrete sampled points along the trajectories. Namely, one does not need to approximate the time derivatives and then identify the continuous-time systems. Moreover, it contributes to a more accurate identification, in comparison with finite-difference based methods, for instance, SINDy\cite{meng2024resolvent,klus2020data}.

\subsection{Learning the Lyapunov Function}

We consider a Lyapunov function candidate of the form $V(x) = \Zk_N(x) \theta$ with $\theta \in \R^N$. Then, to solve the Lyapunov equation~\eqref{eq: lyapunov} and the Zubov equation~\eqref{eq: zubov} respectively, we need the Lyapunov function candidate to satisfy: 
    \begin{equation}\label{eq: Lyapunov_dd}
     \Zk_N(x)\Lb\theta  = -\eta(x),
     \end{equation}
and
   \begin{equation}\label{eq: zubov_dd}
     [\Zk_N(x)\Lb - \eta(x) \Zk_N(x)] \theta  = -\eta(x),
 \end{equation}
respectively. Then, it appears that the problem is reduced to finding the vector $\theta$ such that the least-squares error between the l.h.s. and r.h.s. of \eqref{eq: Lyapunov_dd} or \eqref{eq: zubov_dd} is minimized at the sampled initial conditions. However, directly solving the linear PDEs in such a way may not return the true solutions, especially for the Zubov equation~\eqref{eq: zubov_dd}. For instance, it is evident that $W(x) = 1$ is a trivial solution to~\eqref{eq: zubov}. 

Next, we select a set of boundary points $\{y^{(p)}\}_{p=0}^{P-1}$
and formulate the least-squares problem for solving the Zubov equation as: 
\begin{align} \label{eq: loss}
\theta
  = \mathop{\arg \min}_{K \in \mathbb{R}^N} & \frac{1}{M}\sum_{i=1}^{M} \big\| \left[\Zk_N(x_i)\Lb - \eta(x_i) \Zk_N(x_i)\right] K + \notag \\
  & \eta(x_i) \big\|^2 + \lambda_b \frac{1}{P}\sum_{i=1}^{P} \big\|\Zk_N(y_i) K - b(y_i)\big\|^2,
\end{align}
where $\lambda_b > 0$ is a weight parameter, and $b(y_i)$ is boundary values for the sampled $y_i$. Notably, solving the Lyapunov function~\eqref{eq: Lyapunov_dd} represents a simplified version of this process, and thus the formulation is skipped. For the positive definite function $\eta$, one common choice is $\eta(x) = r|x|^2$, where $r$ is a positive constant. Unless otherwise specified, we will use it with $r = 0.1$ in all the following numerical experiments.
The boundary condition for the Zubov equation consists of both $W(\textbf{0}) = 0$ and $W(y) = 1$ for $y \notin \D$ (and $y \in \partial \D$ if the knowledge is available), whereas the one for the Lyapunov equation is a single point $V(\textbf{0}) = 0$. Note that the positive definiteness of the Lyapunov function $V$ can be ensured by solving the PDEs accurately with the chosen $\eta$~\cite{liu2023physics}.

\begin{rem} \label{rem:identified}
Note that with the identified vector field $\tf(x)$, the stability certificates can also be attained by directly solving the Lyapunov equation $\nabla V(x;\theta)\cdot \tilde{f}(x) = -\eta(x)$ and the Zubov equation $\nabla W(x;\theta)\cdot \tilde{f}(x) = -\eta(x) (1-  W(x;\theta))$, using the method proposed in~\cite{liu2023physics}. However, the proposed method provides a more efficient approach, requiring only the solution of a least-squares problem by repeatedly using the same (observable) test functions.
\end{rem}

The process of identifying the dynamics and computing the Lyapunov functions using the Zubov equation is summarized in Algorithm \ref{alg:Lyapunov}.

\begin{algorithm}[ht]
\SetAlgoLined
\KwIn{X, Y}

\KwOut{$V$ or $\hf \& V$}

Learn $L$ using \eqref{eq:generator} \;

\If{Identifying the dynamics is needed}{
    Identify the dynamics $\hf$ with $L$ using \eqref{eq: f_hat} \;
    Learn the Lyapunov function $V$ with \eqref{eq: loss} \;
}
\Else{ 
    Learn the Lyapunov function $V$ with \eqref{eq: loss} \;
}

\caption{Koopman generator-based Lyapunov function and system identification}
\label{alg:Lyapunov}
\end{algorithm}

\subsection{The Choice of the Dictionary}
In this work, we primarily employ two types of observable test functions widely recognized in the literature~\cite{meng2024resolvent, mauroy2019koopman}: monomials and shallow neural networks, featuring randomly initialized weights and biases.

It is well-known that one-hidden layer neural networks are able to approximate any continuous function on a compact set to any desired degree of accuracy \cite{cybenko1989approximation}. More importantly, as illustrated in~\cite{huang2006extreme, zhou2024physics}, the extreme learning machine (ELM) algorithm can solve the linear stability-related PDEs well using single-hidden layer feedforward neural networks. It returns the Lyapunov functions of form $V(x;\theta) = \sigma(\omega x+b)\theta$,  where $\omega$ and $b$ are randomly generated weights and bias, while $\theta$ is learned by performing a least-squares task. In this paper, we mainly consider the hyperbolic tangent function $\tanh$ as the activation function $\sigma$. 
That said, in order to simultaneously learn the vector field and the Lyapunov function via the Koopman generator, we set $\Zk_N = [\tanh(\omega x + b)^T ; x]$, where $\omega \in \R^{(N-n) \times n}, b \in \R^{N-n}.$ Consequently, we have an approximation of the vector field: $\hf(x) = \Xi [\tanh(\omega x + b)^T ; x]^T$ with $\Xi = L[ (N-n):end, :]$, and $V(x) = [\tanh(\omega x + b)^T ; x] \theta$.

The above selection of observable test functions applies to general nonlinear systems. However, with specific information about the systems, tailored test functions can improve approximation accuracy. For example, in polynomial systems, monomials can be used for more accurate results. Note that with the information, we potentially can get a better approximation of the Koopman generator, and thus better ROA estimates can be attained. In general, the more accurate the resulting generator and corresponding identified vector field are, the closer the ROA estimate can be to the actual ROA.

\section{Numerical Experiments}

We present two numerical examples to demonstrate the effectiveness of the proposed method. For both examples, we assume that the systems are unknown and set $Q = I$ when solving for a quadratic Lyapunov function $V_P$, as shown in Section~\ref{sec:stability}. Recalling equations \eqref{eq: Lyapunov_dd} and \eqref{eq: zubov_dd}, the Zubov equation is in general more challenging to solve and yields larger ROA estimates. Hence we focus exclusively on its results in this section. The Lyapunov equation is a simpler case in terms of solving the linear PDE and can be readily computed by simplifying the process, mainly the residual term and the boundary conditions, with the provided code.
The code is available at \url{https://github.com/RuikunZhou/Unknown_Zubov_Koopman}.

\subsection{Reversed Van der Pol Oscillator}

Consider the reversed Van der Pol oscillator 
$\dot \xb_1 =-\xb_2, \quad 
\dot \xb_2 = \xb_1 - (1 - \xb_1^2)\xb_2, $
with $x:=[x_1, x_2]$. Since this is a polynomial system, we utilize the monomials as the $N=J\times K$ basis functions with $\zk_{i}(x) = x_1^{p}x_2^q$, $p=\mod{(i/J)}$, and $q = \lfloor i/K \rfloor$. In order to solve the Zubov equation \eqref{eq: zubov_dd} sufficiently well using polynomials, we choose $J = K = 8$. In this case, for generator learning, we only need to randomly sample $M = 100$ initial conditions on $\X = [-1.2, 1.2]^2$ and the data points along the trajectories with a sampling frequency $\gamma = 50 Hz$ for $\taus = 5s$. In comparison, the feedforward neural network-based method in~\cite{zhou2022neural} needs 9 million samples with the exact values of the corresponding time derivative for each data point.
As a natural result of the generator approximation $L$, we have the identified vector field $\tf$. Also, we verified that Assumption~\ref{ass: linearization} indeed holds in this case.

Next, we randomly sample 3000 data points across $ [-2.5, 2.5]\times [-3.5, 3.5]$ with 100 boundary points. Note that the boundary conditions given here do not need to be exactly on $\partial \D$, and the samples for generator learning can be reused here. We then solve the Zubov equation~\eqref{eq: zubov_dd} using the learned $L$ with the help of the well-developed toolbox for finding and verifying Lyapunov functions, LyZNet~\cite{liu2024tool}, resulting in a polynomial Lyapunov function. It is remarkable that differing from the regular Lyapunov conditions defined in Theorem~\ref{thm: lyapunov}, we verify the modified Lyapunov conditions for the Lie derivative as~\eqref{eq: relaxed V}, using the integrated satisfiability modulo theories (SMT) solvers in the toolbox. The resulting Lyapunov function and the corresponding ROA are included in Fig.~\ref{fig:vdp_Lyapunov}. 

The parameters for learning the generators and verification for stability are detailed in Table~\ref{tbl:vdp}. We assume that the Lipschitz constant $K_f$ is known. For $\hf$ and $V$, we have the exact expressions, given that they are linear combinations of the observable test functions. We thus can leverage symbolic computations and then compute $K_{\hf}$ and $\nu$ with the help of the SMT solver, dReal~\cite{gao2013dreal} by performing interval analysis with the expressions. Note that when the 2-norm needs to be evaluated, it is over-approximated by the Frobenius norm for the expressions. 
It is noteworthy that while data is sampled on $\X$ to learn the dynamics within the domain of attraction, the certified ROA extends beyond this so-called valid region described in~\cite{zhou2022neural}. This extension is made possible by incorporating physical insights—specifically, leveraging the knowledge that the system is polynomial. With that, we have a more accurate approximation with improved generalization performance outside $\X$. More importantly, both the system identification and stability certification processes require significantly less computational time compared to existing methods. The accuracy and efficiency also result in a substantially smaller $\beta$ in~\eqref{eq: relaxed V}, reduced by a factor of $\frac{1}{10}$ for verification. This reduction is critically important for practical applications involving real-world systems.

\begin{table*}[!ht]
  \caption{Parameters in the Van der Pol Oscillator case, where the first four pertain to learning the Koopman generator, while the last six correspond to those in Proposition~\ref{prop:unknown_stability}} 
  \label{tbl:vdp}
  \centering
  \break
  \begin{tabular}{cccccccccccccc}
    \toprule
    $\taus (s)$ & $\gamma$ & $\mu$ & $\lambda$ & $\lambda_b$ & $K_f$   & $K_{\hf}$ & $\delta$ & $\alpha$ & $\nu$ & $\beta$  \\
      \midrule
    5 & 50 & 2.5 & 1e8 & 100 & 4.90 & 4.90 & 3e-4 & 4.16e-6 & 7.07e-1 & 2.08e-3 \\
    \bottomrule
  \end{tabular}
\end{table*}
\begin{figure} [ht!]
    \centering
    \includegraphics[width=\linewidth]{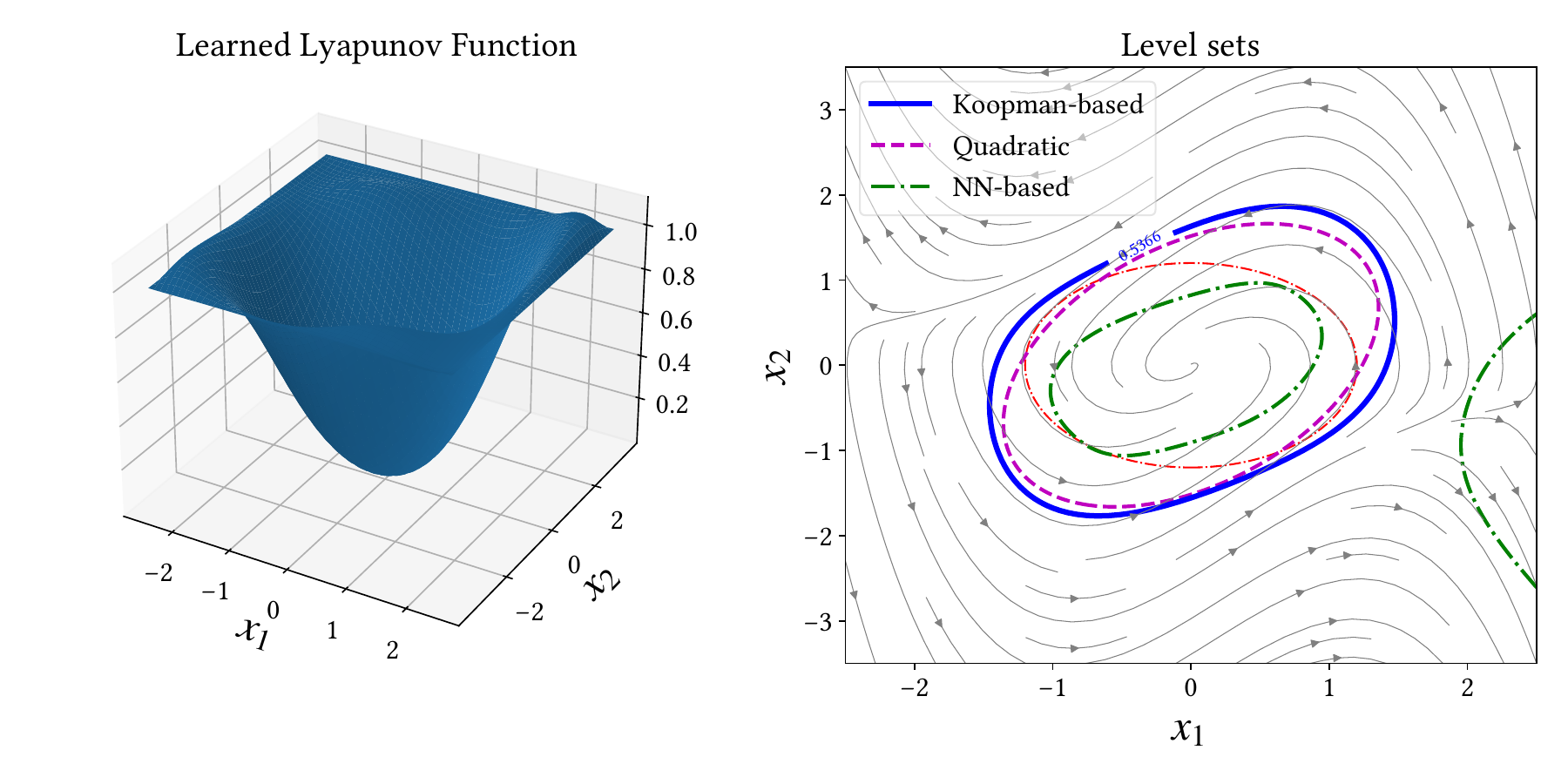}
    \caption{The learned Lyapunov function and corresponding certified ROA estimates, where the blue curve is the one using the proposed method. The magenta dashed one is the verified largest ROA estimate with $V_P(x)$, while the green dot-dashed line is with the neural network-based approach proposed in~\cite{zhou2022neural}. The red dot-dashed circle denotes $\X$ on which we collect the data for computing the Koopman generator.}
    \label{fig:vdp_Lyapunov}
\end{figure}

\subsection{Two-machine Power System}
Consider the two-machine power system in \cite{vannelli1985maximal} with the following governing equation: $\dot{\xb}_1 = \xb_2, \;\dot{\xb}_2 =  -0.5\xb_2 - (\sin(\xb_1 +a)-\sin(a)), $
where $a = \frac{\pi}{3}$. In this example, the system is non-polynomial with a trigonometric function, and we assume no prior knowledge of its dynamics. The dictionary is constructed using 100 $\tanh$-activated neural networks and the state variable $x$. We sample $M = 50^2$ initial conditions within $\X = [-1, 1]^2$, which lies entirely within the true domain of attraction. Following the same procedure as in the previous example, we solve the Zubov equation over $ [-2, 3]\times [-3, 1.5]$. The parameter settings are summarized in Table~\ref{tbl:two machine}, and the resulting Lyapunov function and corresponding ROA estimate can be found in Fig.~\ref{fig:two_machine_Lyapunov}. 

\begin{table*}[!ht]
  \caption{Parameters for the two-machine power system using $\tanh$-activated neural networks} 
  \label{tbl:two machine}
  \centering
  \break
  \begin{tabular}{cccccccccccccc}
    \toprule
    $\taus (s)$ & $\gamma$ & $\mu$ & $\lambda$ & $\lambda_b$ & $K_f$   & $K_{\hf}$ & $\delta$ & $\alpha$ & $\nu$ & $\beta$  \\
      \midrule
    5 & 10 & 3 & 1e8 & 100 & 1.52 & 1.52 & 1e-4 & 2.72e-4 &  1.45 & 8.34e-4 \\
    \bottomrule
  \end{tabular}
\end{table*}

\begin{figure} [ht!]
    \centering
    \includegraphics[width=\linewidth]{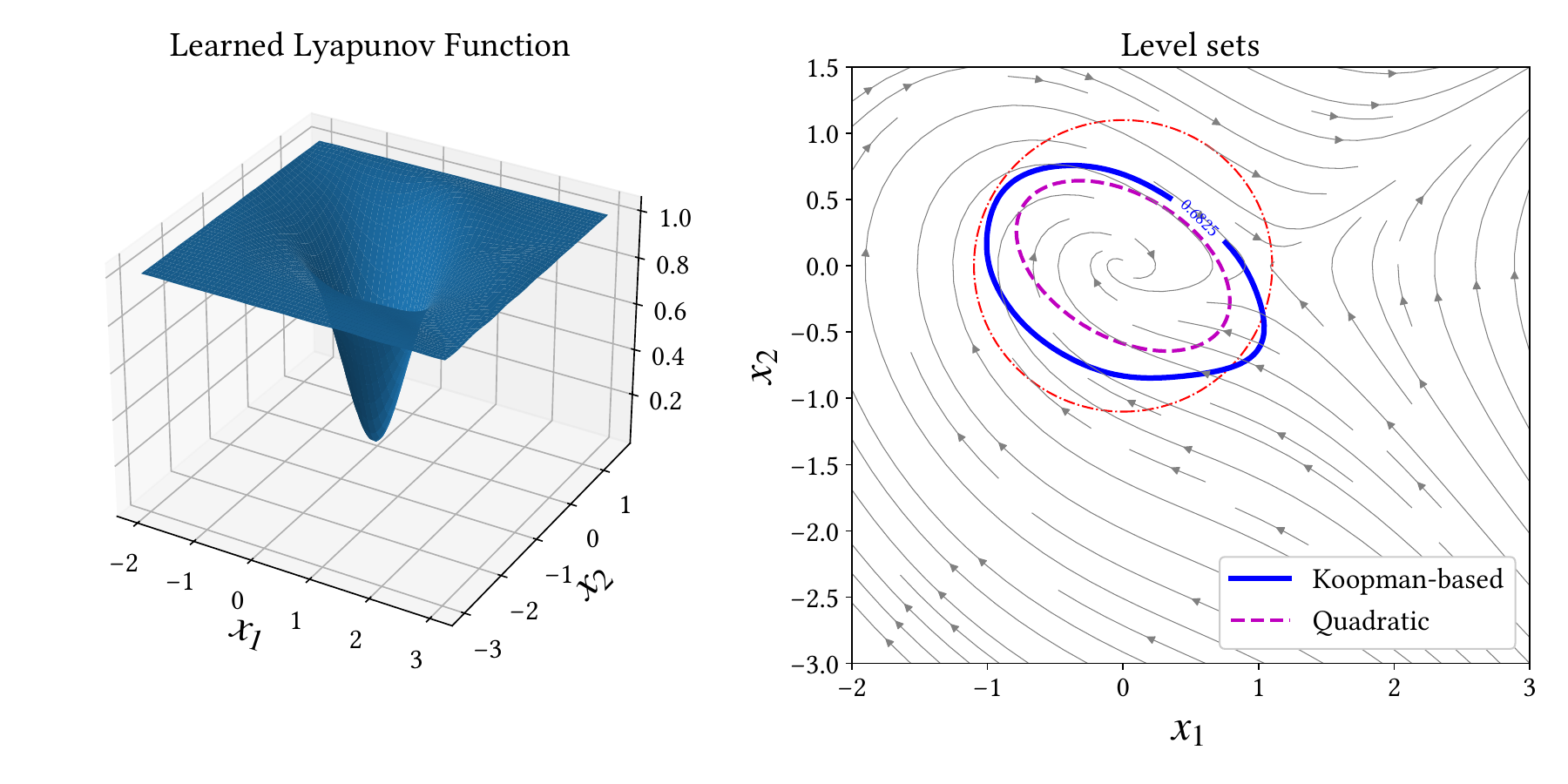}
    \caption{The learned Lyapunov function and corresponding certified ROA estimates, where the red dot-dashed circle denotes $\X$ on which we collect the data for computing the Koopman generator. The blue curve is the certified ROA estimate using the proposed method, while the magenta dashed line is the one with $V_P(x)$.}
    \label{fig:two_machine_Lyapunov}
\end{figure}

We emphasize that the Koopman generator-based method proposed in this work serves as a modular framework for identifying continuous-time systems and finding Lyapunov functions. After identifying the system, one can also directly solve the PDEs to find a Lyapunov function, as discussed in Remark~\ref{rem:identified}. On the other hand, if the vector field is identified using other techniques, such as SINDy, the learned generator $L$ is exclusively used for solving the stability-related PDEs. Notably, leveraging the generator for both system identification and stability analysis significantly reduces computational overhead, underscoring the efficiency of the proposed approach.

\section{CONCLUSION}

In this paper, we propose a framework for simultaneously identifying the vector field and finding a Lyapunov function for an unknown nonlinear system by computing the Koopman generator. Using a shared dictionary of observable test functions for both computing the Koopman generators and solving stability-related PDEs, we establish a more efficient way for constructing stability certificates from data. Leveraging SMT solvers further allows for less conservative ROA estimates, compared to the quadratic Lyapunov functions and an existing neural network-based method. While the current numerical results focus on low-dimensional systems, future work will explore scalability to high-dimensional dynamics with advanced verification tools.

\bibliographystyle{plain}
\bibliography{l4dc2025}

\end{document}